\documentclass[preprint,12pt]{elsarticle}

\makeatletter
\def\ps@pprintTitle{%
  \let\@oddhead\@empty
  \let\@evenhead\@empty
  \let\@oddfoot\@empty
  \let\@evenfoot\@oddfoot
}
\makeatother

\usepackage[ruled,vlined]{algorithm2e}

\usepackage{amsmath}
\usepackage{tikz} 
\usepackage{array}
\usepackage{xcolor,soul}

\usepackage{amsmath,amssymb,amsfonts}
\usepackage{color,soul}
\usepackage{float}
\usepackage{algorithmic}
\usepackage{graphicx}
\usepackage{textcomp}
\def\BibTeX{{\rm B\kern-.05em{\sc i\kern-.025em b}\kern-.08em
    T\kern-.1667em\lower.7ex\hbox{E}\kern-.125emX}}

\usepackage{amsthm,amscd,bm}

\newcommand{\N}{\ensuremath{\mathbb{N}}\xspace}

\newcommand{\myvec}[1]{\ensuremath{\boldsymbol{#1}}\xspace}

\newcommand{\x}{\myvec{x}}

\newcommand{\ZZ}{\mathbb{Z}}

\newcommand{\Hom}{\mathrm{Hom}}

\newcommand{\Ker}{\mathrm{Ker}}
\newcommand{\End}{\mathrm{End}}

\newcommand{\Imma}{\mathrm{Im}}

\newcommand{\glorule}{\ensuremath{{\cal F}}}
\newcommand{\HH}{\ensuremath{{\cal H}}}

\usepackage{xspace, comment, mathdots}

\newtheorem{theorem}{Theorem}
\newtheorem{definition}{Definition}
\newtheorem{remark}{Remark}
\newtheorem{example}{Example}
\newtheorem{lemma}{Lemma}

\newtheorem{corollary}{Corollary}

\newtheorem{observation}{Remark}

\journal{Journal of Computer and System Sciences}

\begin{document}

\begin{frontmatter}

\title{Decidability and Characterization of\\ Expansivity for Group Cellular Automata}

\author[liceo]{Niccol\`o Castronuovo}
\ead{niccolo.castronuovo@studio.unibo.it}

\author[unimib]{Alberto Dennunzio\corref{cor}}
\ead{alberto.dennunzio@unimib.it}

%\author[UCA]{Enrico Formenti}
%\ead{enrico.formenti@univ-cotedazur.fr}

%\author[mfo]{Darij Grinberg}
%\ead{darijgrinberg@gmail.com}

\author[unibo]{Luciano Margara}
\ead{luciano.margara@unibo.it}
\cortext[cor]{Corresponding author}

\affiliation[liceo]{organization={Liceo ``A. Einstein''},%Department and Organization
            %addressline={ciao}, 
            city={Rimini},
            postcode={47923}, 
            %state={Rimini},
            country={Italy}}

%\address[liceo]{Liceo ``A. Einstein'', 47923 Rimini, Italy}

\affiliation[unimib]{organization={Dipartimento di Informatica, Sistemistica e Comunicazione,
  Università degli Studi di Milano-Bicocca},%Department and Organization
            addressline={Viale Sarca 336/14}, 
            city={Milano},
            postcode={20126}, 
            %state={Rimini},
            country={Italy}}
            
%\address[unimib]{Dipartimento di Informatica, Sistemistica e Comunicazione,
  %Università degli Studi di Milano-Bicocca,
  %Viale Sarca 336/14, 20126 Milano, Italy}
  
%\address[UCA]{Universit\'e C\^ote d'Azur, CNRS, I3S, France}

%\address[mfo]{Mathematisches Forschungsinstitut Oberwolfach, Schwarzwaldstr. 9-11, 77709 Oberwolfach-Walke, Germany}

\affiliation[unibo]{organization={Department of Computer Science and Engineering, University of Bologna, Cesena Campus},
            addressline={Via dell'Universita 50}, 
            city={Cesena},
            postcode={47521}, 
            %state={Rimini},
            country={Italy}}
            
%\address[unibo]{Department of Computer Science and Engineering, University of Bologna, Cesena Campus, Via dell'Universit\`a 50, Cesena, Italy}

%\begin{document}
%\maketitle

% \history{Date of publication xxxx 00, 0000, date of current version xxxx 00, 0000.}
% \doi{10.1109/ACCESS.2023.0322000}

% \title{Cellular Automata on Linear Groups
% }\author{
% \uppercase{Niccol\`o Castronuovo}\authorrefmark{1},
% \uppercase{Alberto Dennunzio}\authorrefmark{2}, 
% \uppercase{Luciano Margara}\authorrefmark{3},
% }

% \address[1]{Liceo ``A. Einstein'', Rimini, 47923, Italy (e-mail: castronuovoniccolo@gmail.com)}
% \address[2]{Dipartimento di Informatica, Sistemistica e Comunicazione, Universit\`a 
%   degli Studi di Milano-Bicocca,
%   Viale Sarca 336/14, 20126 Milano, Italy (e-mail: alberto.dennunzio@unimib.it)}

% \address[3]{Department of Computer Science and Engineering, University of Bologna, Cesena Campus, Via dell'Universit\`a 50, Cesena, Italy (e-mail: luciano.margara@unibo.it)}
% \tfootnote{``This work was partially supported by}

% \markboth
% {Dennunzio \headeretal: Cellular Automata on Matrix Groups}
% {Dennunzio \headeretal: Cellular Automata on Matrix Groups}

% \corresp{Corresponding author: Alberto Dennunzio (e-mail: alberto.dennunzio@unimib.it).}

 \begin{abstract}
Group cellular automata are continuous, shift-commuting endomorphisms of $G^\ZZ$, where $G$ is a finite group. We provide an easy-to-check   characterization of expansivity for group cellular automata on abelian groups and we prove that expansivity is a decidable property for general (non-abelian) groups.
 Moreover, we show that the class of expansive group cellular automata is strictly contained 
 in that of topologically transitive injective group cellular automata.
\end{abstract}

%%Research highlights

 % \begin{keywords}
 % Cellular Automata, Group Cellular Automata, Dynamical Behavior, Linear Groups, Matrix Groups. \end{keywords}

% \titlepgskip=-21pt

%\noindent {\bf Keywords:  Cellular Automata, Group Cellular Automata, Dynamical Behavior, Chaos, Decidability}  

%\noindent {\bf MSC2020:} ?? (primary); ?? (secondary).
\end{frontmatter}
%\end{document}

\section{Introduction} \label{introduction}
%=========================================================================
%=========================================================================
%=========================================================================

Cellular Automata (CAs) are discrete-time dynamical systems consisting of a regular grid of variables, each of them taking values from a finite set called alphabet. The 
%global state or 
configuration of a CA is a snapshot of the  values of all variables at a certain time $t$ and it evolves in discrete time steps by means of  a given local rule defining the CA. The local rule updates in a synchronous and homogeneous way each variable of the grid on the basis of the values at time $t-1$ of its neighboring variables (for an introduction to CAs theory, see for example~\cite{hedlund69, Kari2005Survey}).

CAs have been the subject of significant research and are successfully applied in various fields, including computer science, physics, mathematics, biology, and chemistry, for purposes such as simulating natural phenomena, generating pseudo-random numbers, processing images, analyzing universal computation models, and cryptography (see for example~\cite{AnghelescuIS07,MARTINDELREY20051356,DennunzioFGM21INS,kari2000,RubioEWRS04}).

One of the central challenges in CAs theory is describing the global behavior of a CA based on the analysis of its local rule. While the local rule has a finite representation (e.g., a finite table), the global behavior can encompass an arbitrarily large, potentially infinite, amount of information. In fact, the grid of variables whose values define a CA configuration may be infinite and the desired global behavior might only emerge after an arbitrarily large number of time steps.

Many properties related to the temporal evolution of general CAs have been proved to be undecidable (see Section~\ref{literature} for an overview of the main ones). 
Since in practical applications one needs to know if the CA used for modelling a certain system exhibits some specific property, this can be a severe issue. 

Fortunately, the undecidability issue of dynamical properties of CAs can be tackled by placing specific constraints on the model. In many cases, like the one we are exploring in this paper, the alphabet and the local rule are restricted to being a finite group and a homomorphism, respectively, giving rise to group CAs (GCAs). We stress that  these constraints do not at all prevent such CAs from being effective 
%at all hinder the effectiveness of such CAs 
in practical applications. In fact, GCAs can exhibit much of the complex behaviors of general CAs and they are often used in various applications (see for example~\cite{DennunzioFGM21INS,NandiKC94, RubioEWRS04}). 

 Over the past few decades, considerable effort has been devoted to study the decidability of the dynamical behavior  of CAs  defined on various algebraic structures, ranging from cyclic to general finite groups. 
See Section~\ref{literature} for an overview of the main results on this topic.

Expansivity is one of the most important and widely studied dynamical properties (see for example~\cite{Hiraide1987,KitchensSchmidt1989,Lind1982ExponentialRecurrence,Morales2019,Reddy1983}). 
It applies only to reversible systems. A reversible discrete-time dynamical system is said to be expansive if every pair of distinct configurations eventually separate under iterations of the system, either in the future or in the past. An expansive behavior of a system is a hallmark of chaotic dynamics, as it implies that even a small error in the description of the initial state of the system grows over time, making an accurate long-term prediction impossible. A classical example of an expansive system is the shift map in symbolic dynamics.

To our knowledge, there are no significant results on expansive GCAs.
This is due to the technical difficulties of studying reversible systems and the complexity that arises when dealing with non-abelian groups. Some preliminary findings exist for GCAs over $\ZZ_m$~\cite{ChangChang2016}, along with more general—but less conclusive—results concerning automorphisms of compact  groups~\cite{Aoki1975,KitchensSchmidt1989,Lind1982ExponentialRecurrence,Shah2020}. 
%=========================================================================
%=========================================================================
%=========================================================================
%\section{Results Overview}\label{results}
\paragraph{\textbf{Results Overview}}
%=========================================================================
%=========================================================================
%=========================================================================
%In this paper we study expansivity of GCAs. 
%\todo{mettere i teoremi ad ogni risultato che dicimo di dimostrare}
The main results of this paper can be summarized as follows.
\begin{enumerate}
    \item[$-$] A  GCA $\glorule$ on a group $G$ is expansive if and only if
a finite number of simpler and computable GCAs obtained decomposing $\glorule$ and $G$ are expansive. Some of them are defined on abelian groups, while  others, if any, are defined on products of  non-abelian isomorphic  simple groups (Theorem~\ref{rem}).
    \item[$-$] A GCA $\glorule$  defined on an abelian group $G$ is expansive if and only if the GCA $\HH=\glorule-\glorule^{-1}$ is positively expansive (Theorem~\ref{su_ab}).
Since positive expansivity is an easily checkable property on  abelian groups and  the inverse of a GCA is easily computable, we conclude that  expansivity is an easily checkable property as well.
    \item[$-$] For GCAs defined on products of  non-abelian isomorphic simple groups, 
expansivity is equivalent to topological transitivity  (Theorem~\ref{simpleiso}).
    \item[$-$] Expansivity for GCAs is a decidable property (Theorem~\ref{expansive}).
    \item[$-$] Expansive GCAs are topologically transitive (Theorem~\ref{espimplicatrans}), although the converse does not hold  (Example~\ref{ex:1}).
\end{enumerate}
The rest of this paper is organized as follows.

Section~\ref{literature} surveys the existing results regarding the decidability and computability of the main dynamical properties and functions across various classes of CA. In Section~\ref{cellular_a} we recall the fundamental definitions and key results concerning CAs and GCAs, 
including the basic group theory notions and results needed throughout the paper.  Section~\ref{expabel} is devoted to the characterization of expansivity for GCAs on abelian groups, while Section~\ref{expgen} contains the results concerning expansivity for GCAs on general groups. Finally, in Section~\ref{conclu} we draw our conclusions and perspectives. 
%=========================================================================
%=========================================================================
%=========================================================================
\section{Decidability and Characterization of dynamical properties}\label{literature}
%=========================================================================
%=========================================================================
%=========================================================================

Let $C$ be a possibly infinite set of CA, such as the set of GCA.
Let $P$ be a property that a CA may or may not satisfy, such as surjectivity or topological transitivity.
$P$ is  decidable for $C$ if and only if there exists an algorithm that, given a CA $\glorule \in C$, returns ``Yes" if $F$ satisfies $P$, and ``No" otherwise. 
If, instead of a property $P$, we consider a numerical function $N: C \to \mathbb{R}$, such as topological entropy or Lyapunov exponents, $N$ is computable for $C$ if and only if there exists an algorithm that, given a CA $\glorule \in C$, computes $N(\glorule)$.
Deciding a property (or computing a function) involves a computational cost that the notion of decidability (or computability) does not take into account. We will therefore say that a property is efficiently decidable (or a function is efficiently computable) if the algorithm that decides whether the property holds (or computes the function) is efficient, where efficient usually means polynomial-time.
As for CA, efficient algorithms typically analyze the structure of the local rule, which is a finite object, whereas inefficient algorithms usually operate on the space-time dynamics of the CA, which is potentially infinite in size.
The literature contains a number of results related to the decidability and computability of properties and functions, respectively,  across various classes of CA. 
Below, we list some of the most significant ones. 
In what follows, we will denote by  $D$-CA the class of $D$-dimensional CA. The same notation  will also be used  for GCA.\\

\noindent
{\bf General CAs.}\\
\noindent
- Every non-trivial property of limit sets of   1-CA is undecidable~\cite{Kari94}.\\
- Surjectivity and injectivity are decidable for 1-CA~\cite{amoroso1972decision}  and undecidable for  2-CA~\cite{kari1994reversibility}.\\
- Topological entropy for  1-CA is  uncomputable~\cite{Hurd_Kari_Culik_1992}. \\
- Topological entropy  is  efficiently computable for  positively expansive CA~\cite{DamicoMM03}.
\\
- Sensitivity to the initial conditions and topological transitivity are undecidable for  1-CA (even when restricting to the case of reversible 1-CA)~\cite{Lukkarila10}.\\

\noindent
{\bf GCAs.}\\
-  Surjectivity and injectivity are  decidable for $D$-GCA with $D\geq 1$~\cite{BeaurK24}.\\
- Strongly transitive and positively expansive GCA do not exist unless the underlying group is abelian~\cite{nostro}.\\
- Sensitivity to the initial conditions and equicontinuity
are decidable for $D$-GCA with $D\geq 1$~\cite{BeaurK24}.\\
- Non-transitivity  is semi-decidable for $D$-GCA with $D\geq 1$~\cite{BeaurK24}.\\

\noindent
{\bf GCAs on abelian groups.}\\
\noindent
- Sensitivity to the initial conditions, equicontinuity, topological transitivity, ergodicity, positive expansivity, and denseness of periodic orbits are efficiently decidable for 1-GCA on abelian groups~\cite{
DennunzioFGM2020TCS,
Dennunzio20JCSS,
DennunzioFMMP19,
DennunzioFormentiMargara2023,
DBLP:journals/isci/DennunzioFM24,
kari2000}.\\ 

\noindent
{\bf GCA on $\ZZ/m\ZZ$.}\\
\noindent
- Topological entropy  is efficiently computable for  1-GCA~\cite{DamicoMM03}. \\
- Strong transitivity is  efficiently decidable  for $D$-GCA  with  $D\geq 1$~\cite{ManziniM99}.\\
- Lyapunov exponents are efficiently computable for 1-GCA~\cite{FinelliMM98}.

\section{Preliminaries} \label{cellular_a}
%=========================================================================
%=========================================================================
%=========================================================================

In this section, we review the fundamental definitions and key results related to CAs and GCAs. For additional definitions and results, we refer the reader to those introduced in~\cite{GCA24}.

\subsection{General CAs}
Let $G$ be a finite set. A CA \emph{configuration} is any function from $\ZZ$ to $G$, i.e., an element of $G^\ZZ$.  Given a configuration $c\in G^\ZZ$ and any integer $i\in\ZZ$,   $c_i$ denotes the value of $c$ at position $i$, while for any $i,j\in\ZZ$ with $i\leq j$,  $c_{[i,j]}$ denotes the word $c_i\cdots c_j\in G^{j-i+1}$ and $c_{[i,+\infty)}$ denotes the infinite word $c_i c_{i+1}\cdots$.
The set $G^\ZZ$ is also a topological space with the  prodiscrete topology, i.e., the product topology when each factor $G$ is given the discrete topology. 
\begin{comment}
%%%%%CILINDRI 
%%%%%
For any $i,j\in\ZZ$ with $i\leq j$ and any $u\in G^{j-i+1}$, the  cylinder  $C([i,j],u)$ 
is the subset of $G^\ZZ$ defined as $C([i,j],u):= \{c\in G^\ZZ: c_{[i,j]}=u\}$.
%
The cylinders form a clopen basis for the prodiscrete topology and, 
\textbf{SERVONO I CILINDRI??}
\end{comment}
When equipped with that topology,  $G^\ZZ$ turns out to be a compact, Hausdorff, and totally disconnected topological space. Moreover,  $G^\ZZ$ is a Polish space, i.e., a separable completely metrizable topological space. 
Indeed, the set $G^\ZZ$ can be equipped with the standard Tychonoff distance $d$ 
defined as 
$$
\forall c,c^\prime\in G^\ZZ:\  d(c,c')=\begin{cases}
0 & \text{ if }  c=c' \\
2^{
-
%\Delta(c,c') 
\min\{|j|: j\in\ZZ \text{ and } c_j\neq c'_j\}} &\text{otherwise}
\end{cases} 
$$
%where $\Delta(c,c')=\min\{|j|: j\in\ZZ \text{ and } c_j\neq c'_j\}$ 
and the topology induced by the Tychonoff distance coincides with the prodiscrete topology. Since it has no isolated points, the set $G^\ZZ$ is also a Cantor space. 

A \emph{CA} on (the alphabet) $G$ is any continuous function $\glorule: G^\ZZ \to G^\ZZ$ which is also shift commuting, i.e., $\glorule\circ \sigma=\sigma \circ \glorule$, where the shift map $\sigma: G^\ZZ\to G^\ZZ$ is defined as follows
 $$\forall c\in G^\ZZ\;\; \forall i\in \ZZ:\  \sigma(c)_i=c_{i+1}.
 $$

Any CA  can be equivalently defined by means of a  local rule $f:G^{2\rho +1} \to G$, where $\rho\in \N$~\cite{hedlund69}. Namely, a CA $\glorule$ with local rule  $f$ is defined as follows:
$$\forall c \in G^\ZZ\;\; \forall i\in \ZZ:\ \glorule(c)_i=f(c_{i-\rho} ,\dots, c_{i+\rho}).$$ 
The natural ${\rho}$, also denoted by $\rho(\glorule)$, will be said to be  the \emph{radius}  of the CA $\glorule$.

A CA $\glorule$ is said to be \emph{injective} (respectively, \emph{surjective}) if the map $\glorule$ is injective (respectively, surjective).
 We recall that injective CAs are also surjective, and, hence, bijective. Moreover, the inverse of any injective CA is a CA, too, and for this reason injective CAs are sometimes referred to as \emph{reversible}. 
 %A CA is surjective if and only if every configuration has a finite and uniformly bounded number of preimages~\cite{hedlund69}.
%A CA $\glorule$ is said to be  open if $\glorule$ is an open map with respect to the prodiscrete topology, i.e., if it maps open sets to open sets. \textbf{SERVE OPENNESS E NUMERO FINITO DI PREIMMAGINI??}
\smallskip

%We are now ready to define topological transitivity and expansivity.
We now recall the definitions of the dynamical properties that are investigated in this paper with particular emphasis on the first one. 
\begin{definition}\label{def:expansivity}
    An injective (reversible) CA  $\glorule$ is  \emph{expansive}  if 
there exists a constant $\epsilon > 0$  such that for any two configurations  $c, c'  \in G^\ZZ$, there exists an integer $n \in \ZZ$ such that
$
d\big(\glorule^n(c), \glorule^n(c')\big) > \epsilon.
$\\
%The value of the involved $\epsilon$ is called expansivity constant of the CA.
Similarly, a (possibly non injective)  CA $\glorule$ is  \emph{positively expansive} if the previous condition  holds but with the difference that the involved integer $n$ is required to be positive. 
\end{definition} 

\begin{definition}\label{def:transitivity}
    A CA  $\glorule$ is  \emph{topologically transitive}  if for any pair of nonempty open subsets $U,V\subseteq G^\ZZ$ there exists a natural $n\in \N$  such that $\glorule^n(U)\cap V\neq\emptyset$.
\end{definition} 
%
%\textbf{Mixing e altre mixing/ergodic properties}
%
\subsection{Background on Groups and Group CAs}
Let $G$ be a group with identity element $e$. Unless otherwise stated, the operation will be the multiplication. As usual, if the operation is commutative then the group is said to be \textit{abelian}.
A set \( H \subseteq G \)
is a \textit{subgroup} of \( G \), denoted by $H\leq G$, if $H$ forms a group with the same operation of \( G \). The set $\{e\}$ is called the \emph{trivial} group and is a subgroup of any group.   The \textit{centralizer} of a subset $S\subseteq G$ is the subgroup 
$
C_G(S) = \{ x \in G : xg = gx \; \text{ for all } g \in S\}
$.
%If $A$ and $B$ are subsets of the group $G$, $AB$ denotes the set $\{a\cdot b : a\in A;\,b\in B\}$. 

A subgroup \( N \) of \( G \) is said to be \textit{normal} (denoted by \( N \trianglelefteq G \)) if for all \( g \in G \) and \( n \in N \), it holds that \( g n g^{-1} \in N \), where $g^{-1}$ is the inverse element of $g$. Let \( N \) be a normal subgroup of \( G \). The \textit{quotient group} is the set 
$G/N = \{ gN : g \in G \}$, where each element \( gN = \{ g n : n \in N \}\) is called \emph{left coset} of \( N \) in \( G \) and it is also denoted by $[g]$. 
The operation on \( G/N \) is the coset multiplication:
$
(gN)(hN) = (gh)N 
$
for all $g, h \in G$.
% Since \( N \) is normal, the product of two cosets is well-defined. The identity element in \( G/N \) is the coset \( N \) (the coset containing the identity element of \( G \)), and the inverse of \( gN \) is \( g^{-1}N \).
The group $G$ is called \textit{simple group} if it is a nontrivial group and it has no proper nontrivial normal subgroups.
%while $G$ is said to be a \textit{quasi-simple group} if it is a perfect group such that \( G/\Z_G \) is a simple group.

Given two groups \( G \) and \( H \), their \textit{direct product} is the group \( G \times H \) consisting of ordered pairs \( (g, h) \), where \( g \in G \) and \( h \in H \), with the group operation defined component-wise:
$
(g_1, h_1) \cdot (g_2, h_2) = (g_1 \cdot g_2, h_1 \cdot h_2). 
$ 
The identity element of \( G \times H \) is \( (e_G, e_H) \), where \( e_G \) and \( e_H \) are the identity elements of \( G \) and \( H \), respectively.

A \textit{homomorphism} between two groups \( G \) and \( H \) is a map \( \varphi: G \to H \) such that for all \( a, b \in G \),
$
\varphi(a \cdot b) = \varphi(a) \cdot \varphi(b).
$ The set of all homomorphisms from  \( G \) to \(H\) is  denoted \( \operatorname{Hom}(G,H) \).  
A bijective homomorphism \( \varphi: G \to H \) is called \textit{isomorphism} and \( G \) and \( H \) are said to be \textit{isomorphic} (denoted \( G \cong H \)). 
An \textit{endomorphism} 
%(resp., \textit{automorphism}) 
is a homomorphism 
%(resp., isomorphism) 
from a group to itself and 
$\End(G)$ 
%and $Aut(G)$ stand for the sets of 
stands for the set of all the endomorphisms 
%and automorphisms 
of a group $G$.
%, respectively. 
%The set \( \operatorname{Aut}(G) \) forms a group under composition of functions.

 %This is equivalent to saying that \( g N g^{-1} = N \) for all \( g \in G \).

The \textit{kernel} of a homomorphism \( \varphi: G \to H \) is the set 
%\[
$
Ker(\varphi) = \{ g \in G : \varphi(g) = e_H \},
$
%\]
where \( e_H \) is the identity element of \( H \). The kernel is a normal subgroup of \( G \). The homomorphism $\varphi$ is said to be \emph{trivial} iff $Ker(\varphi)=G$. As usual, $\Imma(h)$ denotes the image of $h$. 

%A subgroup $H$ of a group $G$ is \textit{characteristic} if $\varphi(H)\leq H$ for every $\varphi\in Aut(G)$.
%As an example, the center of a group is a characteristic subgroup. 
%A stronger property is  fully invariance. 
A \textit{fully invariant} subgroup of a group $G$ is a subgroup $H$ of $G$ such that, for every $\phi \in \End(G)$, $\phi(H)\leq H$.
%\begin{comment}
In other terms, the restriction of an endomorphism of $G$ to a fully invariant subgroup $H$ induces an endomorphism on $H$.
%\end{comment}
We recall that a fully invariant subgroup is normal. 
%and the commutator subgroup of any group is a fully invariant subgroup.
A group $G$ is called invariantly simple if its only fully invariant 
subgroups are the trivial subgroup and $G$ itself.  
It is known that every invariantly simple finite group is a direct product of isomorphic simple groups~\cite{nostro}.

\smallskip

Let $G$ be now a \emph{finite} group with identity element $e$. The set $G^\ZZ$ is also a group, with the component-wise operation defined by the group operation of $G$, and we denote by $e^\ZZ$ the configuration taking the value $e$ at every integer position, i.e., $e^\ZZ$ is the identity element of the group $G^\ZZ$. Clearly, when equipped with the prodiscrete topology, $G^\ZZ$ turns out to be both a profinite and Polish group. A configuration $c\in G^\ZZ$ is said to be \emph{finite} if the number of positions $i\in\ZZ$ such that $c_i\neq e$ is finite.

If $H\leq G$, then $H^\ZZ$ is a closed subgroup of $G^\ZZ$. Moreover, $H\trianglelefteq G$ if and only if $H^\ZZ\trianglelefteq G^\ZZ$.
In this case, the prodiscrete topologies on $H^\ZZ$ and $(G/H)^\ZZ$ agree with the subspace topology on $H^\ZZ$ and with the quotient topology on $(G/H)^\ZZ$, respectively.  

A CA $\glorule: G^\ZZ\to G^\ZZ$ is said to be a \emph{group CA (GCA)} if $\glorule$ is an endomorphism of $G^\ZZ$. 
In that case, by~\cite[Lemma 3]{SaloT12}, the local rule of $\glorule$ is a homomorphism $f:G^{2\rho+1} \to G$.  %(see~\cite{CR22} for a proof as far as an arbitrary algebraic structure is concerned).
%The  kernel of a GCA  $\glorule$ is $Ker(\glorule)=\{c\in G^\ZZ: \glorule(c)=e^\ZZ\}$. \textbf{SERVE??}
Furthermore, given any function $f:G^{2\rho+1}\to G$, by~\cite[Lemma 8]{SaloT14} (see also the later~\cite[Thm. 1]{GCA24}), it holds that $f\in \Hom(G^{2\rho+1},G)$ if and only if there exist $2\rho+1$ endomorphisms $h_{-\rho},\ldots,h_{\rho}\in \End(G)$, such that $f(g_{-\rho},\ldots,g_{\rho})=h_{-\rho}(g_{-\rho}) \cdots h_{\rho}(g_{\rho})$ for all $g_{-\rho},\ldots,g_{\rho}\in G$ and 
$\Imma(h_i)\subseteq C_G(\Imma(h_j))$ for every $i\neq j$.
In this case, we will write $f=(h_{-\rho},\ldots,h_{\rho})$. Notice that some of the $h_i$'s could be trivial but we will always assume, if not otherwise stated, that at least one between $h_{-\rho}$ and $h_{\rho}$ is non-trivial. 
%and the natural ${\rho}$, also denoted by $\rho(\glorule)$, will be said to be  the \emph{radius}  of the GCA $\glorule$.  

A GCA with local rule $f$  is called  shift-like  if $\rho\geq 1$ and only one between the $h_i$'s is non trivial, while  it is called  identity-like  if $\rho=0$. Surjective shift-like GCA are topologically transitive. 
%
%\textbf{DIRE CHE PER I GCA top trans=mixing=...}
%
\section{Expansivity on abelian groups}
%========================================================
%========================================================
%========================================================
\label{expabel}

%Let $G$ be a finite group, and let $e$ denote its identity element.
%Let $\glorule$ be a GCA on $G$ and $c\in G^\ZZ$ be any configuration. 
We define the following two subsets of $G^\ZZ$ for any finite group $G$.
\begin{eqnarray*}
    L(k) &=&\{ c\in G^\ZZ:\ c_k\neq e \text{ and }  c_i=e \text{ for } i>k \}\\
    R(k) &=&\{ c\in G^\ZZ:\ c_k\neq e \text{ and }  c_i=e \text{ for } i<k \} 
\end{eqnarray*}

Sometimes we will refer to elements of $L(k)$ as  left fronts and to  elements of $R(k)$ as right fronts.   
The position of an element $c$ in $L(k)$ or $R(k)$, denoted by $pos(c)$, is $k$.  
A potential ambiguity in the notation could arise when referring to the position of a configuration that belongs to both $L(k)$ and $R(k')$ (i.e., a finite configuration). Since this situation never occurs, we will avoid overloading the notation by introducing separate left and right positions.
 
The following remark is an immediate consequence of the definition of expansivity and 
the additivity of $\glorule$.
\begin{comment}
%%%%%%%%%%%%%%%%%%
\begin{observation}\label{exp_vs_front}
Let  $\glorule$ be any injective GCA on a group $G$. It holds that 
$\glorule$ is expansive if and only if 
\begin{eqnarray}
 &&   \forall k \in \ZZ \;\; \forall c \in L(k) \;\; 
\exists n \in \ZZ :\ 
pos\bigl(\glorule^{n}(c)\bigr) > k \text { and } \label{eq001}\\
 &&   \forall k \in \ZZ \;\; \forall c \in R(k) \;\; 
\exists n \in \ZZ :\ 
pos\bigl(\glorule^{n}(c)\bigr) < k. \label{eq002}
\end{eqnarray}
\end{observation}
\textbf{ALTERNATIVA}
%%%%%%%%%%%%%
\end{comment}
\begin{observation}\label{exp_vs_front}
Let  $\glorule$ be any injective GCA on a finite group $G$ and let $k\in\ZZ$. It holds that  
$\glorule$ is expansive if and only if 
\begin{eqnarray}
 &&    \forall c \in L(k) \;\; 
\exists n \in \ZZ :\ 
pos\bigl(\glorule^{n}(c)\bigr) > k \text { and } \label{eq001}\\
 &&   \forall c \in R(k) \;\; 
\exists n \in \ZZ :\ 
pos\bigl(\glorule^{n}(c)\bigr) < k. \label{eq002}
\end{eqnarray}
\end{observation}
A counterpart of the equivalence stated in  Remark~\ref{exp_vs_front} holds as far as positive expansivity and any (possibly non injective) GCA $\glorule$ are concerned, with the difference that the integer $n$ involved in both conditions~\eqref{eq001} and~\eqref{eq002} is required to be positive. 

Since $G^\ZZ$ is a compact space, we are able to prove the following Lemma.
\begin{comment}
 \begin{lemma}\label{lyapunov}
Let  $\glorule$ be any injective GCA on a group $G$.  It holds that 
$\glorule$ is expansive if and only if there exists a $k_{\glorule} > 0$ such that:
\begin{eqnarray}
 &&   \forall k \in \ZZ \;\; \forall c \in L(k) \;\; 
\exists n \in \{-k_{\glorule}, \dots , k_{\glorule}\} :\ 
pos\bigl(\glorule^{n}(c)\bigr) > k \text { and } \label{eq1}\\
 &&   \forall k \in \ZZ \;\; \forall c \in R(k) \;\; 
\exists n \in \{-k_{\glorule}, \dots , k_{\glorule}\} :\ 
pos\bigl(\glorule^{n}(c)\bigr) < k. \label{eq2}
\end{eqnarray}
\end{lemma}
\textbf{ALTERNATIVA}
\end{comment}
\begin{lemma}\label{lyapunov}
Let  $\glorule$ be any injective GCA on a finite group $G$ and let $k\in\ZZ$. It holds that 
$\glorule$ is expansive if and only if there exists a natural $k_{\glorule} > 0$ such that:
\begin{eqnarray}
 &&    \forall c \in L(k) \;\; 
\exists n \in \{-k_{\glorule}, \dots , k_{\glorule}\} :\ 
pos\bigl(\glorule^{n}(c)\bigr) > k \text { and } \label{eq1}\\
 &&    \forall c \in R(k) \;\; 
\exists n \in \{-k_{\glorule}, \dots , k_{\glorule}\} :\ 
pos\bigl(\glorule^{n}(c)\bigr) < k. \label{eq2}
\end{eqnarray}
\end{lemma}

\begin{proof}
The ``only if''   implication is trivial. We now prove the ``if'' implication.\\
As we will see, the existence of this value $k_{\glorule}$ follows from the compactness of $G^\ZZ$.

For the sake of contradiction,  suppose that for every $k_{\glorule}>0$ condition~\eqref{eq1} is false (the same reasoning applies assuming condition~\eqref{eq2} to be false).
Let $c\in L(k)$ and let 
$$m(c)=\min\{i\in\N:\ \max\{ pos(\glorule^{-i}(c)),\dots, pos(\glorule^{i}(c)) \}>k\}$$
or, equivalently,
$$m(c)=\min\{i\in\N:\ \max\{ pos(\glorule^{-i}(c)), pos(\glorule^{i}(c)) \}>k\}.$$

Since $\glorule$ is expansive, for every $c\in L(k)$ we have $m(c)< \infty$.
Since $\eqref{eq1}$ is false for every $k_{\glorule}>0$, we can find a  sequence $s=(c^{(i)})_{i\in \N}$ of configurations in $L(k)$  such that
for every $i \geq 1$ we have $m(c^{(i+1)}) > m(c^{(i)})$.
Since $G^\ZZ$ is a compact space, we can extract from $s$ a subsequence  $(c^{(i_j)})_{j\in \N}$ converging to some $c$.  Since for every $j\in \N$, $c^{(i_j)} \in L(k)$, $c$ will belong to $L(k)$ as well.
We now prove that $m(c)$ cannot be finite and then $\glorule$ is not expansive.
By the definition of $m(c)$,  either
$pos(\glorule^{m(c)}(c))>k$  or
$pos(\glorule^{-m(c)}(c))>k$. 
Without loss of generality assume that $pos(\glorule^{m(c)}(c))>k$.

Since $c^{(i_j)}$ converges to $c$ and $m(c^{(i_j)})$ is a strictly increasing sequence, for every $h\in\N$  
we can always pick 
an element $\overline{c}$ in the sequence  $(c^{(i_j)})_{j\in \N}$ such that
\begin{equation*} 
    m(\overline{c})> m(c) \text{  and  } \overline{c}_{[k-h,k]}=c_{[k-h,k]}
\end{equation*}
that leads to 
a contradiction.
\end{proof}
\begin{comment}
We proved  Lemma~\ref{lyapunov} for GCAs, but it is not hard to  extend it to general CAs. 
Also note that, since CAs are shift-commuting maps, the initial front positions $k$ in both Observation~\ref{exp_vs_front} and Lemma~\ref{lyapunov} could be taken to be $0$.
\end{comment}
\begin{definition}\label{angular}
Let $\glorule$ be an injective GCA on the group $G$ which also expansive, and let $k_{\glorule}$ denote the constant introduced in Lemma~\ref{lyapunov}. An element $c\in L(k)$ is said to be a \emph{lower angular configuration} for $\glorule$ if  
$$\forall i \in \{0,\dots,k_\glorule \}\;\; \forall j>k:\   \glorule^{-i}(c)_j =e.$$
Analogously, $c\in L(k)$ is an \emph{upper angular configuration} for $\glorule$ if  
$$\forall i \in \{0,\dots,k_\glorule \}\;\; \forall j>k:\   \glorule^{i}(c)_j =e.$$
See Figure~\ref{fig:LRangular} for a graphical representation of a lower and an upper angular configuration. 
\end{definition}
Definition~\ref{angular} can also be extended to configurations 
$c\in R(k)$, following the same approach.

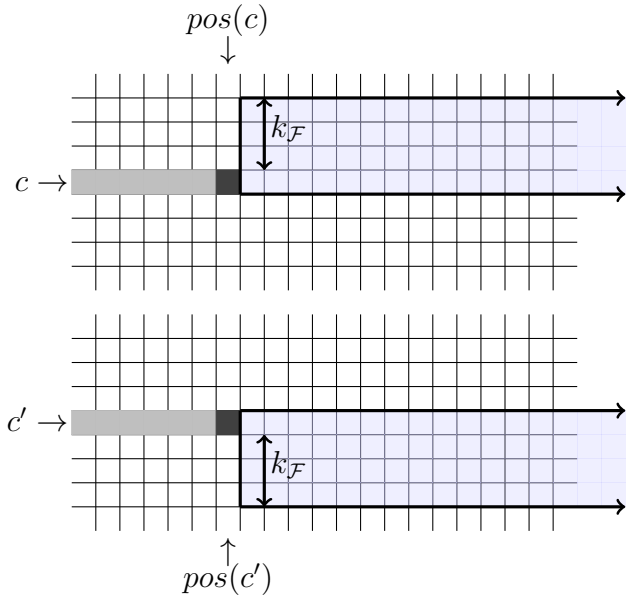
\begin{figure}[ht]
    \centering
\begin{tikzpicture}[scale=0.8]
\definecolor{lightred}{rgb}{1,0.7,0.7}

  % --- griglia 1---
 \def\orx{0}
 \def\ory{0}
 \def\width{.4}
 \def\high{.4}
\def\righe{8}
\def\colonne{20}
\pgfmathsetmacro{\altezza}{\high*(\righe+1)}
\pgfmathsetmacro{\larghezza}{\width*(\colonne+1)}
\foreach \y in {1,...,\righe} {
    \draw (\orx,\ory+\y*\high) -- (\orx+\larghezza,\ory+\y*\high);
  }
\foreach \x in {1,...,\colonne} {
     \draw (\orx+\x*\width,\ory) -- (\orx+\x*\width,\ory+\altezza);
  }

% quadrati grigi
\def\casellay{4}
\pgfmathsetmacro{\cy}{\ory+\casellay*\high}
\def\colore{lightgray}
\foreach \casellax in {0,...,5} {
\pgfmathsetmacro{\cx}{\orx+\casellax*\width}
\filldraw[fill=lightgray, draw=lightgray] (\cx,\cy) rectangle (\cx+\width,\cy+\high);
  }

  % ZERI
\definecolor{lightlightgray}{gray}{0.95} % 0=nero, 1=bianco
\pgfmathsetmacro{\cy}{\ory+\casellay*\high}
\def\colore{lightgray}
\foreach \casellax in {7,...,22} {
\foreach \casellay in {4,...,7} {
\pgfmathsetmacro{\cx}{\orx+\casellax*\width}
\pgfmathsetmacro{\cy}{\ory+\casellay*\high}
\filldraw[fill=blue!20, draw=none, opacity=0.3] (\cx,\cy) rectangle (\cx+\width,\cy+\high);
  }
}

  % quadrato pos
\pgfmathsetmacro{\cx}{\orx+6*\width}
\filldraw[fill=darkgray, draw=darkgray] (\cx,\cy) rectangle (\cx+\width,\cy+\high);

% frecce orizzontali e riga tich verticale
\def\spessore{very thick}
\draw[->,black,\spessore] (\orx+7*\width,\ory+4*\high) -- (\orx+23*\width,\ory+4*\high);
\draw[->,black,\spessore] (\orx+7*\width,\ory+8*\high) -- (\orx+23*\width,\ory+8*\high);

\draw[black,\spessore] (\orx+7*\width,\ory+4*\high) -- (\orx+7*\width,\ory+8*\high);

% scritte
 \node at (-0.5,1.36+\high) {$c \to$};
 \node at (3.0-\width*1,+4) {$\downarrow$};
 \node at (3.0-\width*1,4.5) {$pos(c) $};
% \node at (3.0+\width*18,2.0) {$\leftarrow zeros $};
 \node at (\orx+\width*9,2.7) {$k_\glorule $};
 \draw[<->,black,\spessore] (\orx+8.0*\width,\ory+5*\high) -- (\orx+8.0*\width,\ory+8*\high);

  % --- griglia 1---
 \def\orx{0}
 \def\ory{-4}
 \def\width{.4}
 \def\high{.4}
\def\righe{8}
\def\colonne{20}
\pgfmathsetmacro{\altezza}{\high*(\righe+1)}
\pgfmathsetmacro{\larghezza}{\width*(\colonne+1)}
\foreach \y in {1,...,\righe} {
    \draw (\orx,\ory+\y*\high) -- (\orx+\larghezza,\ory+\y*\high);
  }
\foreach \x in {1,...,\colonne} {
     \draw (\orx+\x*\width,\ory) -- (\orx+\x*\width,\ory+\altezza);
  }

% quadrati grigi
\def\casellay{4}
\pgfmathsetmacro{\cy}{\ory+\casellay*\high}
\def\colore{lightgray}
\foreach \casellax in {0,...,5} {
\pgfmathsetmacro{\cx}{\orx+\casellax*\width}
\filldraw[fill=lightgray, draw=lightgray] (\cx,\cy) rectangle (\cx+\width,\cy+\high);
  }

  % ZERI
\definecolor{lightlightgray}{gray}{0.95} % 0=nero, 1=bianco
\pgfmathsetmacro{\cy}{\ory+\casellay*\high}
\def\colore{lightgray}
\foreach \casellax in {7,...,22} {
\foreach \casellay in {1,...,4} {
\pgfmathsetmacro{\cx}{\orx+\casellax*\width}
\pgfmathsetmacro{\cy}{\ory+\casellay*\high}
\filldraw[fill=blue!20, draw=none, opacity=0.3] (\cx,\cy) rectangle (\cx+\width,\cy+\high);
  }
}

  % quadrato pos
\pgfmathsetmacro{\cx}{\orx+6*\width}
\filldraw[fill=darkgray, draw=darkgray] (\cx,\cy) rectangle (\cx+\width,\cy+\high);

% frecce orizzontali e riga tich verticale
\def\spessore{very thick}
\draw[->,black,\spessore] (\orx+7*\width,\ory+1*\high) -- (\orx+23*\width,\ory+1*\high);
\draw[->,black,\spessore] (\orx+7*\width,\ory+5*\high) -- (\orx+23*\width,\ory+5*\high);

\draw[black,\spessore] (\orx+7*\width,\ory+1*\high) -- (\orx+7*\width,\ory+5*\high);

% scritte
 \node at (-0.57,-4+1.46+\high) {$c' \to$};
 
  \node at (3.0-\width*1,-4.35) {$\uparrow$};
 \node at (3.0-\width*1,-4.8) {$pos(c') $};

% \node at (3.0+\width*18,2.0) {$\leftarrow zeros $};
 \node at (\orx+\width*9,-4+1.1) {$k_\glorule $};
 \draw[<->,black,\spessore] (\orx+8.0*\width,\ory+4*\high) -- (\orx+8.0*\width,\ory+1*\high);

\end{tikzpicture}
 \caption{$c$ is a lower angular configuration and $c'$ is an upper angular configuration. Both are depicted in gray, except for the rightmost element, which is black. Black elements must be different from zero and determine the positions of $c$ and $c'$. All light blue elements must be equal to $e$.
}
    \label{fig:LRangular}
\end{figure}
By the definition of expansivity, if an injective GCA $\glorule$ is expansive, then the position of any left front increases without bound under repeated iterations of either $\glorule$ or $\glorule^{-1}$.  
In the following lemma, we prove that the position of any lower angular configuration increases without bound under repeated iterations of $\glorule$, and that the position of any upper angular configuration increases without bound under repeated iterations of $\glorule^{-1}$.
An analogous result also holds in the case of right fronts.

\begin{lemma}\label{l1}
Let $\glorule$ be an injective GCA on a group $G$. If $\glorule$ is expansive  then the two following properties hold.\\
$(1)$ Let $c\in L(k)$ be an upper angular configuration. Then
$$\forall i<0 \;\; \exists j \in \{ 1,\dots, k_\glorule\}:\ pos(\glorule^{i-j}(c))> pos(\glorule^{i}(c))$$
$(2)$ Let $c\in L(k)$ be a lower angular configuration. Then
$$\forall i>0 \;\; \exists j \in \{ 1,\dots, k_\glorule\}:\ pos(\glorule^{i+j}(c))> pos(\glorule^{i}(c))$$
\end{lemma}
\begin{proof}
We prove only Property $(1)$, since the proof of Property $(2)$ is analogous.
Let $c\in L(k)$ be an upper angular configuration.
By Lemma~\ref{lyapunov} we know that there exists $l\in \{-k_\glorule,\dots k_\glorule \}$  such that $pos(\glorule^l(c)) >k$.  
Since $c$ is an upper angular configuration, $l$ must be negative.

Consider now the ordered set
$A=\{\glorule^{-1}(c),\dots,\glorule^{-k_\glorule}(c)\}$. 
Let $x\in A$   be the first element such that 
 $pos(x)>k$. Since $x$ is the first element of $A$ for which 
 $pos(x)>k$, then $x$ is an upper angular configuration as well.
 Therefore, the theorem can be applied again starting from $c=x$.
 
All elements $y \in A$ that precede $x$ will satisfy $pos(y)\leq k < 
pos(x)$, and therefore the statement is also proved for them.
\end{proof}

In the following lemma, 
we prove that for an injective $\glorule$ which is also expansive,  if the position of a left front $c$ increases, after $n$ iterations of $\glorule^{-1}$, 
by more than a constant threshold (independent of $c$), then $\glorule^{-n}(c)$ cannot be 
a lower angular configuration.
Similarly, if the position of a left front $c$ increases, after $n$ iterations of $\glorule$, 
by more than a constant threshold (independent of $c$), then $\glorule^{n}(c)$ cannot be 
an upper angular configuration.

\begin{lemma}\label{noangular}
Let $\glorule$ be an injective GCA on $G$ which is also expansive.
Let $c\in L(k)$ and  $n\in \N$. If $pos(\glorule^{-n}(c)) > k+k_\glorule\,\rho(\glorule^{-1})$ then
$\glorule^{-n}(c)$ cannot be a lower angular configuration.
\end{lemma}
\begin{proof}

Assume, for contradiction, that $\glorule^{-n}(c)$ is a lower angular configuration
with  $pos(\glorule^{-n}(c)) > k+k_\glorule\,\rho(\glorule^{-1})$.
Let $c_{max}$ be a configuration in  
$\{\glorule^{-i}(c)\}_{0\leq i \leq n}$ having the maximum position.
This ensures that
$pos(c_{max})\geq pos(\glorule^{-n}(c))$. By construction, $c_{max}$  is a lower  angular configuration as well. 
In order to have $\glorule^m(c_{max})=c$ we need that 
\begin{eqnarray*}
    m &\geq& \frac{pos(c_{max})-pos(c)}{\rho(\glorule^{-1})}\\
    &\geq& \frac{pos(\glorule^{-n}(c))-pos(c)}{\rho(\glorule^{-1})}\\
     &>& k_\glorule
\end{eqnarray*}
But $m$ cannot be larger than $k_\glorule$, otherwise $c_{\max}$ would violate 
Lemma~\ref{lyapunov}.
\end{proof}
In the reminder of this section, we will deal with GCAs on abelian groups. For this reason,  for a sake of simplicity, the ``additive'' notation will be employed: the group operation will be denoted by $+$, and so  $-$ will be used as far as an inverse element is involved.  

We are now ready to prove the following theorem, which establishes the equivalence between the expansivity of $\glorule$ and the positive expansivity of $\glorule - \glorule^{-1}$.

\begin{theorem}\label{su_p}
Let $p>1$ be a prime number and $n\in \N$. Let $\glorule$ be any injective GCA on $G=(\ZZ/p\ZZ)^n$. It holds that 
$\glorule$ is expansive if and only if the GCA $\HH=\glorule-\glorule^{-1}$ on $G$ is positively expansive.
\end{theorem}
\begin{proof} 
We are going to show the two implications separately.

We first prove that if $\HH$ is positively expansive, then $\glorule$ is expansive. \\
Suppose instead that $\glorule$ is not expansive.
Then, there exists a configuration $c \in L(k)$
(the same reasoning applies to any $c \in R(k)$ as well) 
such that for every $i \in \ZZ$ we have $pos(\glorule^i(c))\leq k$.
Since for every $m\in \N$
\[
\HH^m(c) = (\glorule - \glorule^{-1})^{m}(c)
= \sum_{j=0}^{m} \binom{m}{j}\,\glorule^{\,m-2j}(c),
\]
we get that $pos(\HH^m(c)) \leq k$. 
Hence, $\HH$ is not positively expansive.
%\textbf{STIAMO USANDO IL CORRISPETTIVO DEL LEMMA 1 PER L'ESPANSIVITA POSITIVA}
\smallskip

We now prove that if 
$\glorule$ is expansive, then $\HH$ is positively expansive.\\
Let $c$ be a configuration in $L(k)$ (the same reasoning applies to $R(k)$ as well).
Let $U(c)$ and $D(c)$ be two sequences of integers defined by
\begin{eqnarray*}
  &&  \forall i\geq 0:\ D(c)_i=pos(\glorule^{i}(c))\text{  and  }\\
   && \forall i\leq 0:\ U(c)_i=pos(\glorule^{i}(c)).
\end{eqnarray*}
If $\glorule$ is expansive, then at least one of the two sequences is unbounded.  
Without loss of generality, suppose that $D(c)$ is unbounded.  
This implies that  there is an infinite ordered set $I(c)$ of positive integers
such that 
\begin{equation*}\label{eq:I(c)}
    \forall i\in I(c)\;\; \forall j\in\{0,\dots,i-1\}:\ D(c)_i > D(c)_j.
\end{equation*}
From the definition of  $I(c)$, it immediately follows that 
$$ \forall l\in I(c) \text{ with } l\geq k_\glorule:\ \glorule^{l}(c) \text{ is a lower angular configuration.}
$$
\noindent
By applying Lemma~\ref{l1}, we can state that
$$
\forall l\in I(c) \text{ with } l\geq k_\glorule:\ I(c)_{l+1}-I(c)_{l}\leq k_\glorule.
$$
In other words, two successive lower angular configurations cannot be more than $k_\glorule$ iterations of $\glorule$ apart.

% refine property $(1)$ by stating that for every $i \geq ang(c)$, at least one of $D_c(i+1), \dots, D_c(i+k_\glorule)$ is strictly greater than $D_c(i)$.

% If $\glorule$ is expansive  at least one of $U_c$ and $D_c$ is unbounded.

For $\HH$ not to be positively expansive, it is necessary that there exists $c \in L(k)$ such that for every %$i \in \ZZ$ \textbf{OCCHIO DOVREBBE ESSERE POSITIVO QUINDI} 
$i \in \N$  we have $pos(\HH^i(c)) \leq k$ and, in particular,
$pos(\HH^{p^\alpha}(c)) \leq k$ for every $\alpha \geq 0$.
Since $\HH^{p^\alpha} = \glorule^{p^\alpha} - \glorule^{-p^\alpha}$, we conclude that 
$$
pos(\HH^{p^\alpha}(c)) \leq k \implies \glorule^{p^\alpha}(c)_{[k+1,+\infty)} = \glorule^{-p^\alpha}(c)_{[k+1,+\infty)}.
$$  
Unfortunately, it is not necessarily the case that $p^\alpha \in I(c)$ or, equivalently, that $\glorule^{p^\alpha}(c)$ is a lower angular configuration.
But, by Lemma~\ref{l1},  we can always find an index $q_\alpha\in I(c)$ such that $|p^\alpha - q_\alpha| \leq k_\glorule$. Since $q_\alpha\in I(c)$, then $\glorule^{q_\alpha}(c)$ is a lower angular 
configuration.

By taking $\alpha$ sufficiently large, we can ensure that $pos(\glorule^{q_\alpha}(c))$ becomes arbitrarily large as well. Since 
$$\glorule^{p^\alpha}(c)_{[k+1,+\infty)} = \glorule^{-p^\alpha}(c)_{[k+1,+\infty)}
\text{  and  }|p^\alpha - q_\alpha| \leq k_\glorule,$$   the configuration $\glorule^{-q_\alpha}(c)$   is guaranteed to be a lower angular configuration.
 However, by Lemma~\ref{noangular}, $\glorule^{-q_\alpha}(c)$ cannot be an angular configuration.
\end{proof}

Theorem~\ref{su_p} can be generalized to abelian groups as follows.

\begin{theorem}\label{su_ab}
Let $\glorule$ be any injective GCA on an abelian group $G$. It holds that 
$\glorule$ is expansive if and only if the GCA $\HH=\glorule-\glorule^{-1}$ on $G$ is positively expansive.
\end{theorem}

\begin{proof} 
For $G=(\ZZ/p\ZZ)^n$ the thesis follows from Theorem~\ref{su_p}.
To extend the result from $G = (\ZZ/p\ZZ)^n$ to general abelian groups, 
it is enough to apply the technique introduced in~\cite[Theorem~2, Corollary~1, and Theorem~3]{DennunzioFormentiMargara2023}.
\end{proof}
\begin{corollary}
Expansivity is an easy-to-check property for injective GCAs on abelian groups. 
\end{corollary}
\begin{proof}
By the results from~\cite{DennunzioFormentiMargara2023} positive expansivity turns out to be an easy-to-check property for GCAs on abelian groups. %Since injectivity is also an easy-to-check property for GCAs on abelian groups, 
Moreover, 
%since establishing if a GCA on an abelian group is 
since the inverse of an injective GCA on an abelian group is easily computable,
by Theorem~\ref{su_ab} we conclude that expansivity is an easy-to-check property for injective GCAs on abelian groups as well.
\end{proof}
\begin{comment}
\begin{remark}
We stress that, by the results from~\cite{DennunzioFormentiMargara2023} positive expansivity turns out to be an easy-to-check property for GCA on abelian groups. Moreover, since the inverse of a GCA on an abelian group is easily computable, we conclude that  expansivity is an easy-to-check property for GCAs on abelian groups as well.
\end{remark}
\end{comment}
\section{Expansivity on general groups}
\label{expgen}
%========================================================
%========================================================
%========================================================
Let $G$ be a finite group and $H$ be any normal subgroup of $G$.
%or every element $g\in G$, the coset  of $g$  in $G/H$, denoted by $[g]$, is defined as
%$[g]=\{gh:\ h\in H\}$. 
%Accordingly,   
For every element $c=(\cdots c_{-1}c_0c_1\cdots)\in G^\ZZ$, the coset  of $c$  in $(G/H)^\ZZ$, denoted by $[c]$, is defined as
$[c]=(\cdots [c_{-1}][c_0][c_1]\cdots)$.

\begin{definition}\label{tilde_barra_def}
Let $G$ be a finite group and let $H \trianglelefteq G$.  
Let $\glorule$ be a GCA on $G$ such that $\glorule(H^\ZZ) \subseteq H^\ZZ$.  
We define the maps 
\begin{align*}
\overline{\glorule}_{H} & : H^\ZZ \to H^\ZZ, \\
\widetilde{\glorule}_{H} & : (G/H)^\ZZ \to (G/H)^\ZZ
\end{align*}
as follows:
\begin{align}
\forall c \in H^\ZZ: \quad & \overline{\glorule}_H(c) := \glorule(c), \label{Fbar} \\
\forall [c] \in (G/H)^\ZZ: \quad & \widetilde{\glorule}_H([c]) := [\glorule(c)]. \label{Ftilde}
\end{align}
\end{definition}

Note that, by Equations~\eqref{Fbar} and~\eqref{Ftilde}, and since $\glorule(H^\ZZ)\subseteq H^\ZZ$,  
the maps $\overline{\glorule}_H$ and $\widetilde{\glorule}_H$ are well-defined GCAs on
$H^\ZZ$ and $(G/H)^\ZZ$, respectively. From now on, when the subgroup $H$ is clear from the context, we will simplify the notation by writing $\overline{\glorule}$ and $\widetilde{\glorule}$ instead of $\overline{\glorule}_{H}$ and $\widetilde{\glorule}_{H}$.

It is straightforward to verify that if $\glorule$ is a GCA on $G$ with local rule $f=(h_{-\rho},\ldots,h_{\rho})$ and $H\trianglelefteq G$, then 
$$
\glorule(H^\ZZ)\subseteq H^\ZZ
  \iff  
f(H^{2\rho+1})\subseteq H 
 \iff  
h_i(H)\subseteq H \ \text{for all } -\rho \leq i \leq \rho.
$$

Moreover, if $f=(h_{-\rho},\ldots,h_{\rho})$ is the local rule of $\glorule$, then the local rules $\overline{f}$ and $\widetilde{f}$ of $\overline{\glorule}$ and $\widetilde{\glorule}$ are given by
$\overline{f}=(h_{-\rho}|_H,\ldots,h_{\rho}|_H)$ and $
\widetilde{f}=(\widetilde{h}_{-\rho},\ldots,\widetilde{h}_{\rho})$,
where $h_i|_H$ denotes the restriction of $h_i$ to $H$, and 
$\widetilde{h}_i([x])=[h_i(x)]$ for all  $x\in G$.

The next theorem will enable us to reduce the study of the expansivity of any GCA to 
the study of the expansivity of two GCAs defined on smaller groups.
%We will use the following theorem proved in Section A
\begin{theorem}\label{expansive}
Let $G$ be a finite group and let $H\trianglelefteq G$. Let $\glorule$ be an injective GCA on $G$ such that $\glorule(H^{\mathbb Z})\subseteq H^{\mathbb Z}$. It holds that $\glorule$ is expansive if and only if both the (injective) GCAs $\widetilde{\glorule}$ and $\overline{\glorule}$ are expansive. 
\end{theorem}
\begin{proof}
First of all, note that, by~\cite[Theorem 2]{nostro}, both the GCAs $\widetilde{\glorule}$ and $\overline{\glorule}$ are injective. The equivalence regarding expansivity follows directly from~\cite[Theorem~A]{GlocknerRaja2017}. 
We point out that in~\cite{GlocknerRaja2017} the authors adopt a slightly different definition of expansivity:  
an automorphism $\alpha$ of $G$ is called expansive if there exists a neighborhood $U$ of $e$ such that  
\[
\bigcap_{n \in \mathbb{Z}} \alpha^n(U) = \{e\}.
\]
Since the prodiscrete topology on $G^\mathbb{Z}$ is induced by the bi-invariant 
Tychonoff metric $d$, this definition of expansivity is equivalent to that provided in  Definition~\ref{def:expansivity} from Section~\ref{cellular_a} (see also~\cite{Shah2020}).
\end{proof}
\begin{comment}
A fully invariant subgroup of a group $G$ is a subgroup $H \leq G$ such that, for every $\phi \in \End(G)$, one has $\phi(H)\leq H$.  
A finite group $G$ is called invariantly simple if its only fully invariant 
subgroups are the trivial subgroup and $G$ itself.  
It is known that every invariantly simple group is a direct product of isomorphic simple groups~\cite{nostro}.  
\end{comment}
We now briefly recall (in a slightly modified form) the decomposition function for GCAs introduced in~\cite{nostro}.

\vspace{0.3cm}
\begin{algorithm}[H]
\SetAlgoNlRelativeSize{0}
\SetAlgoNoLine
\SetKwFunction{Decomposition}{Decomposition}
\SetKwProg{Fn}{Function}{:}{}
\Fn{\Decomposition{$\glorule$,$G$}}
{
\If{$G$ is invariantly simple}
{\KwRet{$\{(\glorule,G)\}$\;}}
\Else{
Let $H$ be a non-trivial proper fully invariant subgroup of $G$\;
a  $\gets$ \Decomposition{$\widetilde{\glorule},G/H$}\;
b  $\gets$ \Decomposition{$\overline{\glorule},H$}\;
\KwRet{$a\cup b$}\;}
}
\end{algorithm}
\vspace{0.3cm}
The function {\tt Decomposition} takes as input a GCA $(\glorule, G)$ and produces 
a finite collection $\{(\glorule_1,G_1),\dots,(\glorule_k,G_k)\}$ of GCAs, where each $G_i$ is either abelian or non-abelian invariantly simple (see~\cite{nostro} for details). Moreover, $\glorule$ is injective iff each $\glorule_i$ is, too.

Note that if $G_i$ is abelian, then it is isomorphic to $(\ZZ/p\ZZ)^n$ for some prime $p$ and some $n\in \N$.  
This follows from the fact that all simple abelian groups are isomorphic to $\ZZ/p\ZZ$ for a prime $p$.

We are now ready to state the following result.

\begin{theorem}\label{rem}
Let $\glorule$ be an injective  GCA on $G$ and let
$$\{(\glorule_1,G_1),\dots,(\glorule_k,G_k)\}={\tt Decomposition}(\glorule,G)\enspace,$$ 
be the collection of the necessarily injective GCAs returned by the function {\tt Decomposition} when invoked with argument $(\glorule,G)$. 
%\textbf{SONO TUTTI INJECTIVE}
%(in such a collection of GCA we are sure that each $\glorule_i$ is injective). 
The following equivalence holds: 
$$(\glorule,G) \text{ is expansive } \iff \forall i\in \{1,\dots,k\}:\ (\glorule_i,G_i)\text{ is expansive.}$$
\end{theorem}
\begin{proof}
The thesis follows from the definition of the {\tt Decomposition} function and from repeated applications of Theorem~\ref{expansive}.
\end{proof}
We recall the following definition (see~\cite{nostro}).
\begin{definition}
A group $G$ that is the product $S_1 \times \cdots \times S_m$ of finite non-abelian isomorphic simple groups $S_i$ is said to be \emph{minimal} with respect to the action of a given GCA $\glorule$ on $G$ if there do not exist two non-empty, disjoint sets $I,J \subseteq \{1,2,\ldots,m\}$ with $I \cup J = \{1,2,\ldots,m\}$ such that
\[
\glorule\left(\left(\prod_{i\in I} S_i\right)^\ZZ\right) \subseteq \left(\prod_{i\in I} S_i\right)^\ZZ   \text{  and  }  
\glorule\left(\left(\prod_{i\in J} S_i\right)^\ZZ\right) \subseteq \left(\prod_{i\in J} S_i\right)^\ZZ.
\]
\end{definition}
Clearly, minimality is a decidable property.

The next theorem plays a crucial role in extending the 
decidability of expansivity from abelian groups to general groups.
\begin{theorem}\label{simpleiso}
Let $G$ be a product of finite non-abelian isomorphic simple groups. Let $\glorule$ be an injective  GCA on $G$. 
%Suppose that $G$ is minimal with respect to $\glorule$. 
It holds that $\glorule$ is topologically transitive if and only if it is expansive. Furthermore, both these properties are decidable.
\end{theorem}
\begin{proof}
If $G$ is not minimal with respect to $\glorule$, then we factor $(\glorule,G)$ into the direct product of a finite number  of GCAs $\{(\glorule_1,G_1),\dots, (\glorule_k,G_k)\}$ 
such that every  $G_i$ is a product of  non-abelian isomorphic simple groups that is also minimal with respect to $\glorule_i$. 
Clearly, $\glorule$ is expansive (respectively, topologically transitive) on $G$ if and only if, for every $i \in \{1, \dots, k\}$, $\glorule_i$ is expansive (respectively, topologically transitive) on $G_i$.
Hence, without loss of generality, we will assume for the remainder of this proof that $G$ is minimal with respect to $\glorule$.

Write $G=S_1\times\cdots\times S_m,$ where the $S_i$'s are finite non-abelian isomorphic simple groups. Denote by $f=(h_{-\rho},\ldots,h_{\rho})$ the local rule of $\glorule$.
Assume that $\glorule$ is surjective, otherwise, 
$\glorule$ would be neither topologically transitive nor expansive and the thesis would be proved. 
By~\cite[Lemma~11]{nostro} we know that
for every $i\in \{-\rho,\dots,\rho\}$ there exists  $J_i \subseteq \{1,\dots,m\}$ such that
$\Imma(h_i)=\prod_{t\in J_i}S_t$. Moreover, there exist two positive integers $K$ and $C$ depending only on $G$ and $f$ such that $\glorule^K=\sigma^{-C\Delta}$ where  
\begin{equation} \label{Delta}
\Delta=\sum_{i=-\rho}^{\rho} i|J_i|.
\end{equation}
Hence, either $\glorule^K$ is the identity map (when $\Delta=0$), or it is a non-negative power of the shift map (when $\Delta \neq 0$).
In the first case, $\glorule$ is neither topologically transitive nor expansive, whereas in the second case, $\glorule$ is both topologically transitive and expansive.

For deciding 
both these properties, it is sufficient to compute the value of $\Delta$ defined in Equation \eqref{Delta}.
\end{proof}

In the next theorem, we show that for injective GCAs, expansivity implies topological transitivity.  
For general CAs, to the best of our knowledge, this has not yet been established.
\begin{theorem}\label{espimplicatrans}
All injective GCAs that are expansive  are also topologically transitive.
\end{theorem}
\begin{proof}
By~\cite[Theorem 4]{nostro}, it is sufficient to prove the thesis for abelian groups and for products of simple non-abelian isomorphic groups.

Let $\glorule$ be an injective GCA on an abelian group $G$. Suppose by contradiction that 
$\glorule$ is not topologically transitive. By~\cite{ShRo88}  we know that either $\glorule$
is not surjective or there exists $n>0$ such that $\glorule^n-I$ is not surjective, where $I$ denotes the identity map.
By hypothesis, since $\glorule$
is injective, it is also surjective.  %then $\glorule$ is not expansive. 
So, there exists %$n\in \ZZ$ \textbf{DEVE ESSERE POSITIVO} 
$n\in \N$ such that $\glorule^n-I$ is not surjective.
By the Garden of Eden Theorem~\cite{moore62,myhill63}, this implies that there exist two finite configurations $x,y \in G^\ZZ$ such that $(\glorule^n-I)(x)=(\glorule^n-I)(y)$  and, hence,
there is a finite configuration $c=x-y$ such that $(\glorule^n-I)(c)=e^\ZZ$ or, equivalently, $\glorule^n(c)=c$. 
%Let $x,y \in G^\ZZ$ such that $x-y=c$. 
Since $c$ is finite and $\glorule^n(c)=c$, for every $\epsilon>0$ we can find a sufficiently large   $j_\epsilon\in \N$ such that 
  $$\forall i\in\ZZ:\ d(\glorule^i(\sigma^{j_\epsilon}(x)),\glorule^i(\sigma^{j_\epsilon}(y)))<\epsilon.$$ This shows that $\glorule$ is not expansive.

In the case of products of simple non-abelian isomorphic groups, the thesis has been already proved in Theorem~\ref{simpleiso}.
\end{proof}

An alternative and completely different approach to proving Theorem~\ref{espimplicatrans} is to show that the traces of expansive GCAs are transitive subshifts of finite type~\cite{kari2025privcom}. 
%(Jarkko Kari, personal communication, August 2025)
\begin{observation}
We stress that, as a consequence of Theorem~\ref{espimplicatrans} and the results from~\cite{nostro} concerning the equivalence of topological transitivity to the main topological mixing and ergodic properties, it also holds that expansive (reversible) GCAs exhibit all those properties. 
\end{observation}
Furthermore, since for  GCAs on simple non-abelian isomorphic groups, expansivity is equivalent to topologically transitivity (Theorem~\ref{simpleiso}), any GCA that is topologically transitive but not expansive must be defined on an abelian group. In the following example, we show that such a GCA indeed exists. Before illustrating it, recall that a Linear CA $\glorule$ over $G=(\ZZ/m\ZZ)^n$ is a GCA  on $G$ defined by $n\times n$ matrices with coefficients in $\ZZ/m\ZZ$ and by a local rule  which is a linear combination  matrices-vectors. In this way,  an $n\times n$ matrix $M_{\glorule}$ of Laurent polynomials in the variables $X$ and $X^{-1}$ and with coefficients in $\ZZ/m\ZZ$ is associated with $\glorule$. Almost all the dynamical properties of $\glorule$ have been characterized in terms of easy-to-check conditions on the characteristic polynomial of $M_{\glorule}$.
\begin{example}\label{ex:1}
%There exists a topologically transitive injective GCA that is not expansive. 
Let $G=(\ZZ/2\ZZ)^2$ and let $\glorule$
%, represented as a matrix, be
be the Linear CA over $G$ having the following matrix associated with it:
%defined by the matrix
$$M_\glorule=\begin{bmatrix}
0&1\\
1&X
\end{bmatrix}.
$$
Since $\det(M_\glorule)=1$, it holds that 
$\glorule$ is injective.
The characteristic polynomial of $M_\glorule$  is $\chi(t)=1+t^2 + tX$. 
Since $\gcd (1+t^2,t)=1$, by the results from~\cite{DBLP:journals/isci/DennunzioFM24} we get that 
$\glorule$ is topologically transitive. Now, the matrix associated with 
$\glorule^{-1}$  %is defined by the matrix 
is 
$$M_\glorule^{-1}=\begin{bmatrix}
X&1\\
1&0
\end{bmatrix}.
$$
We are going to show that $\glorule$ 
is not expansive. 
Let $c \in R(k)$ with $k > 0$. Since neither $M_\glorule$ nor $M_\glorule^{-1}$ contains polynomials in which $X$ appears with negative exponents, we get that 
$
pos(\glorule^i(c)) \ge k$ for every $i \in \mathbb{Z}. 
$ Hence, we conclude that $\glorule$ is not expansive.
\end{example}

Based on the results proved so far, we can now establish the most significant result of this work.
\begin{theorem}\label{main2}
 Expansivity for  GCAs is decidable.
\end{theorem}
\begin{proof}
Let $\glorule$ be an injective GCA on a group $G$.
Let $$\{(\glorule_1,G_1),\dots,(\glorule_k,G_k)\}={\tt Decomposition}(\glorule,G).$$
By Theorem~\ref{expansive},  $\glorule$   is expansive if and only if 
all the GCAs $(\glorule_i,G_i)$  are expansive. 
By Theorem~\ref{su_ab}, if $G_i$ is abelian testing the expansivity of $(\glorule_i,G_i)$ 
is decidable.
By Theorem~\ref{simpleiso}, if $G_i$ is not abelian expansivity is decidable as well. 
\end{proof}

%========================================================
%========================================================
%========================================================
\section{Conclusions and further work}
%========================================================
%========================================================
%========================================================
\label{conclu}

The dynamical behavior of group automorphisms is a central topic in the theory of dynamical systems. In 1975, Nobuo Aoki~\cite{Aoki1975} proved that ergodic automorphisms of compact metric groups are isomorphic to Bernoulli shifts. This result highlights the strong connection between ergodic automorphisms and Bernoulli shifts, which serve as a prototypical example of expansive dynamical systems.
In this paper, we have focused on a specific class of compact metric groups, namely $G^\mathbb{Z}$ with $G$ a finite group, and on a specific class of automorphisms, namely those that are continuous and commute with the shift map, i.e., the injective group cellular automata.
For this class of automorphisms, we proved that expansivity is a decidable property and implies topological transitivity (that is not known to be decidable), which in turn is equivalent to ergodicity. Moreover, we showed that there exist topologically transitive injective GCAs that are not expansive, demonstrating that the isomorphism proposed by Aoki does not preserve expansivity.

Future work includes incorporating into our results a study of the computational cost of the algorithms we propose for deciding the expansivity of GCAs, which was not addressed in the present work. It would also be interesting to extend our analysis to more general algebraic structures beyond groups, and to automorphisms that do not commute with the shift.

\paragraph{Acknowledgements}
This work was partially supported by the PRIN 2022 PNRR project ``Cellular Automata Synthesis for Cryptography Applications (CASCA)'' (P2022MPFRT) funded by the European Union – Next Generation EU, and by the HORIZON-MSCA-2022-SE-01 project 101131549 ``Application-driven Challenges for Automata Networks and Complex Systems (ACANCOS)''.

\bibliographystyle{elsarticle-num}
\bibliography{matrix_groups_bib}

\begin{thebibliography}{10}
\expandafter\ifx\csname url\endcsname\relax
  \def\url#1{\texttt{#1}}\fi
\expandafter\ifx\csname urlprefix\endcsname\relax\def\urlprefix{URL }\fi
\expandafter\ifx\csname href\endcsname\relax
  \def\href#1#2{#2} \def\path#1{#1}\fi

\bibitem{hedlund69}
G.~A. Hedlund, Endomorphisms and automorphisms of the shift dynamical system, Mathematical Systems Theory 3 (1969) 320--375.

\bibitem{Kari2005Survey}
J.~Kari, Theory of cellular automata: A survey, Theoretical Computer Science 334~(1-3) (2005) 3--33.
\newblock \href {https://doi.org/10.1016/j.tcs.2004.11.021} {\path{doi:10.1016/j.tcs.2004.11.021}}.

\bibitem{AnghelescuIS07}
P.~Anghelescu, S.~Ionita, E.~Sofron, Block encryption using hybrid additive cellular automata, in: A.~K{\"{o}}nig, M.~K{\"{o}}ppen, N.~K. Kasabov, A.~Abraham (Eds.), 7th International Conference on Hybrid Intelligent Systems, {HIS} 2007, Kaiserslautern, Germany, September 17-19, 2007, {IEEE} Computer Society, 2007, pp. 132--137.

\bibitem{MARTINDELREY20051356}
A.~M. del Rey, J.~P. Mateus, G.~R. S{\'a}nchez, A secret sharing scheme based on cellular automata, Applied Mathematics and Computation 170~(2) (2005) 1356 -- 1364.

\bibitem{DennunzioFGM21INS}
A.~Dennunzio, E.~Formenti, D.~Grinberg, L.~Margara, Decidable characterizations of dynamical properties for additive cellular automata over a finite abelian group with applications to data encryption, Information Sciences 563 (2021) 183--195.

\bibitem{kari2000}
J.~Kari, Linear cellular automata with multiple state variables, in: H.~Reichel, S.~Tison (Eds.), STACS 2000, Vol. 1770 of LNCS, Springer-Verlag, 2000, pp. 110--121.

\bibitem{RubioEWRS04}
C.~F. Rubio, L.~H. Encinas, S.~H. White, {\'{A}}.~M. del Rey, G.~R. S{\'{a}}nchez, The use of linear hybrid cellular automata as pseudo random bit generators in cryptography, Neural Parallel {\&} Scientific Comp. 12~(2) (2004) 175--192.

\bibitem{NandiKC94}
S.~Nandi, B.~K. Kar, P.~P. Chaudhuri, Theory and applications of cellular automata in cryptography, {IEEE} Trans. Computers 43~(12) (1994) 1346--1357.

\bibitem{Hiraide1987}
K.~Hiraide, Expansive homeomorphisms of compact surfaces are pseudo-anosov, Osaka Journal of Mathematics 27 (1990) 117--162.

\bibitem{KitchensSchmidt1989}
B.~Kitchens, K.~Schmidt, Automorphisms of compact groups, Ergodic Theory and Dynamical Systems 9~(3) (1989) 495--508.

\bibitem{Lind1982ExponentialRecurrence}
D.~Lind, \href{https://doi.org/10.1007/BF02760537}{Ergodic group automorphisms are exponentially recurrent}, Israel Journal of Mathematics 41~(4) (1982) 313--320.
\newblock \href {https://doi.org/10.1007/BF02760537} {\path{doi:10.1007/BF02760537}}.
\newline\urlprefix\url{https://doi.org/10.1007/BF02760537}

\bibitem{Morales2019}
C.~A. Morales, \href{https://doi.org/10.1007/s00574-019-00133-w}{A survey on expansive homeomorphisms}, Boletim da Sociedade Brasileira de Matemática 50 (2019) 41--76.
\newblock \href {https://doi.org/10.1007/s00574-019-00133-w} {\path{doi:10.1007/s00574-019-00133-w}}.
\newline\urlprefix\url{https://doi.org/10.1007/s00574-019-00133-w}

\bibitem{Reddy1983}
W.~Reddy, Pointwise and setwise expansive homeomorphisms, Journal of the London Mathematical Society 27~(2) (1983) 327--336.
\newblock \href {https://doi.org/10.1112/jlms/s2-27.2.327} {\path{doi:10.1112/jlms/s2-27.2.327}}.

\bibitem{ChangChang2016}
C.-H. Chang, H.~Chang, On the bernoulli automorphism of reversible linear cellular automata, Information Sciences 345 (2016) 217--225.
\newblock \href {https://doi.org/10.1016/j.ins.2016.01.062} {\path{doi:10.1016/j.ins.2016.01.062}}.

\bibitem{Aoki1975}
N.~Aoki, Ergodic automorphisms of compact metric groups are isomorphic to bernoulli shifts, Publications math\'ematiques et informatique de Rennes S4 (1975) 1--10.

\bibitem{Shah2020}
R.~Shah, \href{https://nyjm.albany.edu/j/2020/26-15.html}{Expansive automorphisms on locally compact groups}, New York Journal of Mathematics 26 (2020) 285--302.
\newline\urlprefix\url{https://nyjm.albany.edu/j/2020/26-15.html}

\bibitem{Kari94}
J.~Kari, {R}ice's theorem for the limit sets of cellular automata, Theoretical Computer Science 127~(2) (1994) 229--254.

\bibitem{amoroso1972decision}
S.~Amoroso, Y.~Patt, A decision procedure for surjectivity and injectivity of parallel maps for tessellation structures, Journal of Computer and System Sciences 6~(5) (1972) 448--464.

\bibitem{kari1994reversibility}
J.~Kari, Reversibility and surjectivity of cellular automata in dimension d $\geq$ 2, Information and Computation 118~(1) (1994) 192--210.

\bibitem{Hurd_Kari_Culik_1992}
L.~P. Hurd, J.~Kari, K.~Culik, The topological entropy of cellular automata is uncomputable, Ergodic Theory and Dynamical Systems 12~(2) (1992) 255--265.

\bibitem{DamicoMM03}
M.~d'Amico, G.~Manzini, L.~Margara, On computing the entropy of cellular automata, Theoretical Computer Science 290~(3) (2003) 1629--1646.

\bibitem{Lukkarila10}
V.~Lukkarila, Sensitivity and topological mixing are undecidable for reversible one-dimensional cellular automata, Journal of Cellular Automata 5~(3) (2010) 241--272.

\bibitem{BeaurK24}
P.~B{\'{e}}aur, J.~Kari, Effective projections on group shifts to decide properties of group cellular automata, Int. J. Found. Comput. Sci. 35~(1{\&}2) (2024) 77--100.

\bibitem{nostro}
N.~Castronuovo, A.~Dennunzio, L.~Margara, A divide and conquer algorithm for deciding group cellular automata dynamics, \url{https://arxiv.org/abs/2507.09761} (2025).

\bibitem{DennunzioFGM2020TCS}
A.~Dennunzio, E.~Formenti, D.~Grinberg, L.~Margara, Dynamical behavior of additive cellular automata over finite abelian groups, Theoretical Computer Science 843 (2020) 45--56.

\bibitem{Dennunzio20JCSS}
A.~Dennunzio, E.~Formenti, D.~Grinberg, L.~Margara, An efficiently computable characterization of stability and instability for linear cellular automata, J. Comput. Syst. Sci. 122 (2021) 63--71.

\bibitem{DennunzioFMMP19}
A.~Dennunzio, E.~Formenti, L.~Manzoni, L.~Margara, A.~E. Porreca, On the dynamical behaviour of linear higher-order cellular automata and its decidability, Information Sciences 486 (2019) 73--87.

\bibitem{DennunzioFormentiMargara2023}
A.~Dennunzio, E.~Formenti, L.~Margara, An easy to check characterization of positive expansivity for additive cellular automata over a finite abelian group, IEEE Access 11 (2023) 121246--121255.
\newblock \href {https://doi.org/10.1109/ACCESS.2023.3328540} {\path{doi:10.1109/ACCESS.2023.3328540}}.

\bibitem{DBLP:journals/isci/DennunzioFM24}
A.~Dennunzio, E.~Formenti, L.~Margara, An efficient algorithm deciding chaos for linear cellular automata over $(\mathbb{Z}/m\mathbb{Z})^n$ with applications to data encryption, Inf. Sci. 657 (2024) 119942.

\bibitem{ManziniM99}
G.~Manzini, L.~Margara, A complete and efficiently computable topological classification of d-dimensional linear cellular automata over {$\mathbb{Z}_m$}, Theoretical Computer Science 221~(1-2) (1999) 157--177.

\bibitem{FinelliMM98}
M.~Finelli, G.~Manzini, L.~Margara, Lyapunov exponents versus expansivity and sensitivity in cellular automata, Journal of Complexity 14~(2) (1998) 210--233.

\bibitem{GCA24}
A.~Dennunzio, E.~Formenti, L.~Margara, On the dynamical behavior of cellular automata on finite groups, IEEE Access 12 (2024) 122061--122077.

\bibitem{SaloT12}
V.~Salo, I.~T{\"{o}}rm{\"{a}}, On shift spaces with algebraic structure, in: S.~B. Cooper, A.~Dawar, B.~L{\"{o}}we (Eds.), How the World Computes - Turing Centenary Conference and 8th Conference on Computability in Europe, CiE 2012, Cambridge, UK, June 18-23, 2012. Proceedings, Vol. 7318 of Lecture Notes in Computer Science, Springer, 2012, pp. 636--645.

\bibitem{SaloT14}
V.~Salo, I.~T{\"{o}}rm{\"{a}}, Color blind cellular automata, J. Cell. Autom. 9~(5-6) (2014) 477--509.

\bibitem{GlocknerRaja2017}
H.~Glöckner, C.~R.~E. Raja, Expansive automorphisms of totally disconnected, locally compact groups, Journal of Group Theory 20~(3) (2017) 589--619.
\newblock \href {https://doi.org/10.1515/jgth-2016-0051} {\path{doi:10.1515/jgth-2016-0051}}.

\bibitem{ShRo88}
M.~Shirvani, T.~D. Rogers, Ergodic endomorphisms of compact abelian groups, Communications in {M}athematical {P}hysics 118 (1988) 401--410.

\bibitem{moore62}
E.~F. Moore, Machine models of self-reproduction, Proceedings of Symposia in Applied Mathematics 14 (1962) 13--33.

\bibitem{myhill63}
J.~Myhill, The converse to {M}oore's garden-of-eden theorem, Proceedings of the American Mathematical Society 14 (1963) 685--686.

\bibitem{kari2025privcom}
J.~Kari, private communication (Aug. 2025).

\end{thebibliography}

\end{document}